\documentclass[conference,a4paper]{IEEEtran}

\usepackage{latexsym, amssymb,graphicx, amsmath}
\newtheorem{Theorem}{Theorem}
\newtheorem{Definition}{Definition}

\newtheorem{Proposition}{Proposition}
\newtheorem{Remark}{Remark}
\newtheorem{Example}{Example}

\makeatother

\begin{document}

\sloppy

\title{Reconstruction Guarantee Analysis of Binary Measurement Matrices Based on Girth}

 \author{
   \IEEEauthorblockN{Xin-Ji Liu \IEEEauthorrefmark{1}\IEEEauthorrefmark{2}, Shu-Tao Xia\IEEEauthorrefmark{1}\IEEEauthorrefmark{2}}
   \IEEEauthorblockA{\IEEEauthorrefmark{1}Graduate School at Shenzhen, Tsinghua University,
     Shenzhen, Guangdong, China\\
     \IEEEauthorrefmark{2}Tsinghua National Laboratory for Information Science and Technology, Beijing, China\\
     Email: liuxj11@mails.tsinghua.edu.cn, xiast@sz.tsinghua.edu.cn}
 }

\addtolength{\textheight}{2.5cm}

\maketitle
\begin{abstract}
Binary 0-1 measurement matrices, especially those from coding theory, were introduced to compressed sensing (CS) recently.
Good measurement matrices with preferred properties, e.g., the restricted isometry property (RIP) and nullspace property (NSP), have no known general ways to be efficiently checked.
Khajehnejad \emph{et al.} made use of \emph{girth} to certify the good performances of sparse binary measurement matrices.
In this paper, we examine the performance of binary measurement
matrices with uniform column weight and arbitrary girth
under basis pursuit.
Explicit sufficient conditions of exact reconstruction
are obtained, which improve
the previous results derived from RIP for any girth $g$ and results from NSP when $g/2$ is odd.
Moreover, we derive explicit $l_1/l_1$, $l_2/l_1$ and $l_\infty/l_1$
sparse approximation guarantees.
These results further show that
large girth has positive impacts on the performance of binary measurement
matrices under basis pursuit, and the binary parity-check matrices of good
LDPC codes are important candidates of measurement matrices.

\end{abstract}

\section{Introduction}\label{intro}
Consider a \emph{$k$-sparse} signal $\textit{\textbf{x}}=(x_{1}, x_{2}, \ldots, x_{n})^T\in\mathbb{R}^{n}$ which has at most $k$ nonzero entries.
Let $H\in \mathbb{R}^{m\times n}$ be a measurement matrix with $m\ll n$ and $\textit{\textbf{y}}=H\textbf{\textit{x}}$ be the $m$-dimensional measurement vector.
The \emph{compressed sensing} (CS) problem \cite{ecjr,dono} aims to solving the following \emph{$l_{0}$-optimization} problem
\begin{equation}\label{l0opt}
    \min ||\textit{\textbf{x}}||_{0} \quad s.t.\quad H\textbf{\textit{x}}=\textit{\textbf{y}},
\end{equation}
where $||\textit{\textbf{x}}||_{0}\triangleq |\{i:x_{i}\neq 0\}|$ denotes the $l_{0}$-norm or (Hamming) weight of $\textit{\textbf{x}}$.
Unfortunately, it is well-known that the problem (\ref{l0opt}) is NP-hard in general.
In compressed sensing, a convex relaxation of (\ref{l0opt}), the \emph{$l_{1}$-optimization} (a.k.a. basis pursuit, BP), is usually used instead,
\begin{equation}\label{l1}
    \min ||\textit{\textbf{x}}||_{1} \quad s.t.\quad H\textbf{\textit{x}}=\textit{\textbf{y}},
\end{equation}
where $||\textit{\textbf{x}}||_{1}\triangleq\sum_{i=1}^{n}|x_{i}|$
denotes the $l_{1}$-norm of $\textit{\textbf{x}}$.
The optimization problem (\ref{l1}) could be turned into a \emph{linear programming} (LP) problem and thus tractable.

In order to recover the original signals exactly, measurement matrices need to satisfy some properties.
Cand$\grave{\textup{e}}$s and Tao \cite{ectt05} showed that if $H$ satisfies the \emph{restricted isometry property} (RIP) with relatively small restricted isometry constant $\delta_{2k}$, $l_1$-optimization (\ref{l1}) can output exact recovery of any $k$-sparse signal.
In addition, when the signal $\textit{\textbf{x}}$ is not exactly sparse, (\ref{l1}) is supposed to produce a high quality $k$-\emph{sparse approximation} of $\textit{\textbf{x}}$ as good as one measures the $k$ largest values of $\textit{\textbf{x}}$ directly.
RIP is only sufficient, another property, called the \emph{nullspace property} (NSP) \cite{wxbh}, was proposed and proved to be both sufficient and necessary for a measurement matrix to be effective under basis pursuit.
If $H$ satisfies NSP, $l_1$-optimization (\ref{l1}) will produce exact recovery of sparse signals and high quality sparse approximations of approximately sparse signals.
Many random matrices, e.g., Fourier matrices, Gaussian matrices, \emph{etc.}, were verified to satisfy RIP with overwhelming probability.
However, there is no guarantee that a specific realization of random matrices works and it costs large storage space.
As a result, deterministic constructions of measurement matrices are necessary.
Among them, binary 0-1 matrices from coding theory attract many attentions
\cite{aafm,lggz}.
Recently, Dimakis, Smarandache, and Vontobel \cite{adrs} found that LP decoding of LDPC codes is very similar to LP reconstruction of compressed sensing, and they further showed that the sparse binary parity-check matrices of good LDPC codes can be used as \emph{provably} good measurement matrices for compressed sensing under basis pursuit.

For a binary matrix $H$, let $G_{H}$ denote its \emph{Tanner graph} which was introduced to study LDPC codes \cite{tanner}. The \emph{girth} of $H$ or $G_H$ is defined as the minimum length of circles in $G_{H}$. Note that girth is an even number not smaller than 4.
Deterministic matrices with RIP or NSP have no general ways to be definitely constructed or efficiently verified \cite{akas}.
Instead, in \cite{adrs,akas}, girth is used to evaluate the performance of sparse binary measurement matrices under basis pursuit.
It was shown that binary measurement matrices with girth $\Omega(log\,n)$, uniform row weights and uniform column weights have robust recovery guarantees under basis pursuit with high probability.
Since girth is much easier to check, it is considered as a good property to certify good binary measurement matrices.

In this paper, we examine the performance of binary measurement
matrices with uniform column weight $\gamma$ and arbitrary girth $g$
under basis pursuit.
This kind of matrices are often used as measurement matrices, such as those based on expander graphs \cite{sjwx} where binary matrices with uniform column weight are associated with left-regular bipartite graphs.
Let $H$ be such a binary measurement matrix with uniform column weight $\gamma$ and girth $g$.
Suppose any two distinct columns of $H$ have at most $\lambda$ common 1's, i.e., $\lambda$ is the maximum inner product of two different columns.
Obviously, for $g\geq6$, $\lambda=1$.
One of our main results implies that $H$ could exactly recover a $k$-sparse signal if $k<\sum_{u=0}^{t+1} (\gamma-1)^u$, where $t\triangleq\lfloor\frac{g-6}{4}\rfloor\geq0$. When $g=6$, the best known result derived from RIP \cite[Prop.1]{jbsd} or NSP \cite[Lem.12, Th.3]{adrs} implies the exact recovery for $k<(\gamma+1)/2$, while our result only requires $k<\gamma$.
In fact, we derive explicit $k$-sparse approximation guarantees which involve only $k,\;g,\;\gamma,\;\lambda$ when $g=4$ and
$k,\;g,\;\gamma$ when $g\geq6$ respectively.
Our results show that the larger $g$ or/and $\gamma$ are, the larger $k$ will be, thus the better the measurement matrix $H$ is.
This further suggests that good parity-check matrices for LDPC codes are important candidates of measurement matrices for compressed sensing.

\section{Notations and Preliminaries}\label{prelimi}
Let $H\in \mathbb{R}^{m\times n}$ with $m\ll n$, define $[n]\triangleq\{1,\ldots,n\}$ and
\begin{eqnarray*}
Nullsp_{\mathbb{R}}(H)&\triangleq&\{\textbf{\textit{w}}\in
\mathbb{R}^n \;:\; H\textbf{\textit{w}}=\textbf{0}\},\\
Nullsp_{\mathbb{R}}^{*}(H)&\triangleq& Nullsp_{\mathbb{R}}(H)\backslash \{ \textbf{0}\}.
\end{eqnarray*}
Let $K\subseteq[n]$, $|K|=k$, define $\bar{K}\triangleq[n]\backslash K$.
For any real vector $\textit{\textbf{a}}=(a_1, a_2, \ldots, a_n)^T\in \mathbb{R}^n$, $\textit{\textbf{a}}_K$ denotes the vector with $k$ entries of $\textit{\textbf{a}}$ whose positions appear in $K$.
Define the $l_2$ and $l_\infty$ norms of $\textit{\textbf{a}}$ as $||\textit{\textbf{a}}||_2\triangleq\sqrt{\sum_i{a_i^2}}$ and $||\textit{\textbf{a}}||_\infty\triangleq \max_i |a_i|$.
The support of $\textit{\textbf{a}}$ is defined by $supp(\textit{\textbf{a}})\triangleq\{i\,:\, a_i\neq0\}$ and denotes $|\textit{\textbf{a}}|\triangleq(|a_1|,|a_2|,\ldots,|a_n|)^T$.
For $\textit{\textbf{a}}\neq \textbf{0}\in \mathbb{R}^n$,
the AWGN channel pseudoweight of $\textit{\textbf{a}}$ is
$w_{p}^{AWGN}(\textit{\textbf{a}}) \triangleq \frac {||\textit{\textbf{a}}||_{1}^2}{||\textit{\textbf{a}}||_{2}^2}$,
and the max-fractional pseudoweight is
$w_{max-frac}(\textit{\textbf{a}}) \triangleq \frac {||\textit{\textbf{a}}||_{1}}{||\textit{\textbf{a}}||_{\infty}}$.

Firstly, the $l_p/l_q$ sparse approximation guarantees are used here as performance metrics for CS under basis pursuit.
\begin{Definition}\cite{adrs}
An \emph{$l_p/l_q$ sparse approximation guarantee} for basis pursuit means that the $l_1$-optimization (\ref{l1}) outputs an estimate $\hat{\textit{\textbf{x}}}$ such that
\begin{equation}\label{lplq}
  ||\textit{\textbf{x}}-\hat{\textit{\textbf{x}}}||_p\leq C_{p,q}(k)\cdot ||\textit{\textbf{x}}-\textit{\textbf{x}}'||_q,
\end{equation}
where $||\cdot||_p$ ($||\cdot||_q$) denotes the $l_p$ ($l_q$) norm, and $\textit{\textbf{x}}'$ is the best $k$-sparse approximation of $\textit{\textbf{x}}$.
\end{Definition}

Next, we introduce the well-known NSP.
\begin{Definition}\cite{wxbh,adrs}
  Let $H\in \mathbb{R}^{m\times n}$, $k\in \mathbb{N}$, and $C\geq1$.
  We say that $H$ has the \emph{nullspace property} $NSP_{\mathbb{R}}^{\leq}(k,C)$ denoted by $H\in NSP_{\mathbb{R}}^{\leq}(k,C)$, if $\forall K\subseteq[n]$ with $|K|\leq k$
  \begin{equation*}
    C\cdot||\textit{\textbf{w}}_K||_1\leq||\textit{\textbf{w}}_{\bar{K}}||_1,\quad \forall \textit{\textbf{w}}\in Nullsp_{\mathbb{R}}(H).
  \end{equation*}
\end{Definition}

In \cite{adrs}, Dimakis, Smarandache, and Vontobel pointed out that the LP reconstruction of CS is very similar to the LP decoding of LDPC codes.
Performance guarantees in LDPC codes were translated to the corresponding sparse approximation guarantees in CS, e.g., the next result gives the connection between AWGN channel pseudoweight and $l_2/l_1$ guarantee.
\begin{Proposition}\label{citeth13}
\cite[Th.13]{adrs}
Let $H$ be a binary $m\times n$ measurement matrix, $K\subseteq[n]$, $|K|=k$, $C'$ be an arbitrary positive real number such that $C'>4k$.
If $\forall\textit{\textbf{w}}\in Nullsp_{\mathbb R}^{*}(H)$
\begin{equation}\label{awgncond}
  w_p^{AWGN}(|\textit{\textbf{w}}|)\geq C',
\end{equation}
then the output $\hat{\textit{\textbf{x}}}$ produced by the $l_1$-optimization (\ref{l1}) satisfies
\begin{equation}\label{oldl2l1}
||\textit{\textbf{x}}-\hat{\textit{\textbf{x}}}||_2\leq\frac{C''}{\sqrt{k}}||\textit{\textbf{x}}_{\bar{K}}||_1 \quad \mbox{with} \quad C''\triangleq \frac{1}{\sqrt{\frac{C'}{4k}}-1}.
\end{equation}
\end{Proposition}

\section {Main Results}\label{mainresult}
We start this section by giving a sufficient condition for a binary measurement matrix to provide the $l_1/l_1$, $l_2/l_1$ and $l_\infty/l_1$ sparse approximation guarantees for $l_1$-optimization.
\begin{Theorem}\label{preth}
Let $H$ be a binary $m\times n$ measurement matrix such that $\forall \textit{\textbf{w}}=({w_{1}, w_{2}, \ldots, w_{n}})^T\in Nullsp_{\mathbb R}^{*}(H)$ and $\forall i \in supp(\textit{\textbf{w}})$,
\begin{equation}\label{cond}
    |w_{i}|\leq \frac{||\textit{\textbf{w}}||_{1}}
    {C_0},
\end{equation}
then for any set $K\subseteq[n]$ with $|K|=k<\frac{C_0}{2}$, the estimate $\hat{\textit{\textbf{x}}}$ produced by the $l_1$-optimization (\ref{l1}) will satisfy:
\begin{enumerate}
  \item
  \begin{equation}\label{l1l1}
    ||\textit{\textbf{x}}-\hat{\textit{\textbf{x}}}||_1\leq\frac{C_1}{k}||\textit{\textbf{x}}_{\bar{K}}||_1 \quad \mbox{with} \quad C_1\triangleq \frac{C_0}{\frac{C_0}{2k}-1};
  \end{equation}
  \item
  \begin{equation}\label{l2l1}
  ||\textit{\textbf{x}}-\hat{\textit{\textbf{x}}}||_2\leq\frac{C_2}{k}||\textit{\textbf{x}}_{\bar{K}}||_1 \quad \mbox{with} \quad C_2\triangleq \frac{\sqrt{C_0}}{\frac{C_0}{2k}-1};
  \end{equation}
  \item
  \begin{equation}\label{linfl1}
  ||\textit{\textbf{x}}-\hat{\textit{\textbf{x}}}||_\infty\leq\frac{C_3}{k}||\textit{\textbf{x}}_{\bar{K}}||_1 \quad \mbox{with} \quad C_3\triangleq\frac{1}{\frac{C_0}{2k}-1}.
  \end{equation}
\end{enumerate}
\end{Theorem}
\begin{proof}
Since $H\textit{\textbf{x}}=\textit{\textbf{y}}$ and $H\hat{\textit{\textbf{x}}}=\textit{\textbf{y}}$, define $\textit{\textbf{w}}\triangleq \textit{\textbf{x}}-\hat{\textit{\textbf{x}}}$, then $\textit{\textbf{w}}\in Nullsp_\mathbb{R}(H)$.
If $\textit{\textbf{w}}=\textbf{0}$, (\ref{l1l1}), (\ref{l2l1}) and (\ref{linfl1}) holds obviously, so we only need to consider the case $\textit{\textbf{w}}\in Nullsp^*_\mathbb{R}(H)$.
\begin{enumerate}
\item (\ref{cond}) implies $\forall\, \textit{\textbf{w}}\in Nullsp_{\mathbb{R}}^{*}(H), \;||\textit{\textbf{w}}_K||_1\leq\frac{k}{C_0}||\textit{\textbf{w}}||_1$, or $H\in NSP_{\mathbb{R}}^{\leq}(k,C)$ with $C=\frac{C_0}{k}-1$.  Hence, (\ref{l1l1}) follows directly by the nullspace condition \cite[Th.1]{wxbh}, \cite[Th.5]{adrs}.
\item Applying (\ref{cond}) to \cite [Lemma 1]{sxff06}, we have
\begin{equation}\label{awgnineq}
w_p^{AWGN}(|\textit{\textbf{w}}|)=\frac{||\textit{\textbf{w}}||_1^2}
{||\textit{\textbf{w}}||_2^2}\geq C_0,
\end{equation}
which implies $||\textit{\textbf{x}}-\hat{\textit{\textbf{x}}}||_2\leq ||\textit{\textbf{x}}-\hat{\textit{\textbf{x}}}||_1 /  \sqrt{C_0}$ and (\ref{l2l1}) follows by (\ref{l1l1}).
\item
(\ref{cond}) implies $w_{max-frac}(|\textit{\textbf{w}}|)\geq C_0$.
Hence, (\ref{linfl1}) follows directly by \cite[Th.14]{adrs}.
\end{enumerate}
\end{proof}
\begin{Remark}
From (\ref{l1l1}), (\ref{l2l1}) or (\ref{linfl1}), it is easy to see that if $H$ satisfies (\ref{cond}), for any $k$-sparse signal $\textit{\textbf{x}}$ with $k<\frac{C_0}{2}$, $l_1$-optimization (\ref{l1}) outputs exact recovery of $\textit{\textbf{x}}$.
Even when $\textit{\textbf{x}}$ is not exactly sparse, (\ref{l1}) still produces a good $k$-sparse approximation that is within only a factor from the best $k$-term approximation of $\textit{\textbf{x}}$, where $k$ is smaller than $\frac{C_0}{2}$.
\end{Remark}
\begin{Remark}
Since (\ref{awgnineq}) can be implied by (\ref{cond}), Proposition \ref{citeth13} gives that the $l_1$-optimization (\ref{l1}) outputs exact recovery of $k$-sparse signals with $k<\frac{C_0}{4}$, which is only a half of $k<\frac{C_0}{2}$ that our result indicates.
Besides, when $C_0>4k$ one can easily verify that $\frac{C_2}{k}<\frac{C''}{\sqrt{k}}$ by letting $C'=C_0$, which means our approximation guarantee is sharper than that of Proposition \ref{citeth13}.
\end{Remark}

The following two theorems which evaluate $C_0$ in (\ref{cond}) are the key results of this paper. One is for binary matrices with girth $4$ and the other for girth greater than $4$.

\begin{Theorem}\label{thg4}
Let $H$ be a binary $m\times n$ measurement matrix with uniform column weight $\gamma$ and girth $g$.
Suppose the maximum inner product of any two distinct columns of $H$ is $\lambda$.
Then for any $\textit{\textbf{w}}=({w_{1}, w_{2}, \ldots, w_{n}})^T\in Nullsp_{\mathbb R}^*(H)$ and any $i \in supp(\textit{\textbf{w}})$,
\begin{equation}\label{mreg4}
    |w_{i}|\leq \frac{||\textit{\textbf{w}}||_{1}}
    {C_0}\quad \mbox{with}\quad C_0\triangleq\frac{2\gamma}{\lambda}.
\end{equation}
\end{Theorem}
\begin{proof}
Firstly, we claim that for any binary matrix $H$ with uniform column weight $\gamma$ and
any $\textit{\textbf{w}}=(w_{1}, w_{2}, \ldots, w_{n})^T\in Nullsp_{\mathbb{R}}^{*}(H)$, we have
\begin{equation}\label{poseqneg}
  \sum_{j\;:\; w_{j}>0}{w_{j}}=-\sum_{j\;:\; w_{j}<0}{w_{j}}=\frac{||\textit{\textbf{w}}||_{1}}{2}.
\end{equation}
This is because by summing all rows of $H$, we obtain a new row $\tilde{\textit{\textbf{h}}}=(\gamma, \gamma, \ldots, \gamma)$,
and $\tilde{\textit{\textbf{h}}}\textit{\textbf{w}}=0$ implies (\ref{poseqneg}).

For any $\textit{\textbf{w}}\in Nullsp_{\mathbb{R}}^{*}(H)$, we split its support
$supp(\textit{\textbf{w}})$ into $supp(\textit{\textbf{w}}^{+})$ and $supp(\textit{\textbf{w}}^{-})$, where
\begin{eqnarray}
\label{split1}
supp(\textit{\textbf{w}}^{+})&\triangleq&\{\;i\;:\; w_{i}>0\;\},\nonumber\\
\label{split2}
supp(\textit{\textbf{w}}^{-})&\triangleq&\{\;i\;:\; w_{i}<0\;\}.\nonumber
\end{eqnarray}
For any fixed $i\in supp(\textit{\textbf{w}}^{+})$,
by summing the $\gamma$ rows of $H$ each of which has component `1' in the $i$-th position, we obtain a new row $\tilde{\textit{\textbf{h}}}(i)=(\tilde{h}_1(i), \tilde{h}_2(i), \ldots, \tilde{h}_n(i))$ where $\tilde{h}_i(i)=\gamma$ and for any $j\neq i$, $0\leq \tilde{h}_j(i)\leq\lambda$, since the maximum inner product of any two distinct columns of $H$ is $\lambda$.
Clearly, $\tilde{\textit{\textbf{h}}}(i)\textit{\textbf{w}}=0$, i.e.,
\begin{eqnarray*}
0=
\sum_{j\in supp(\textit{\textbf{w}}^+)}\tilde{h}_j(i)w_j+\sum_{j\in supp(\textit{\textbf{w}}^-)}\tilde{h}_j(i)w_j,
\end{eqnarray*}
which implies that
\begin{eqnarray}\label{pfpart1}
-\sum_{j\in supp(\textit{\textbf{w}}^-)}\tilde{h}_j(i)w_j=\sum_{j\in supp(\textit{\textbf{w}}^+)}\tilde{h}_j(i)w_j \overset{(a)}{\ge} \gamma w_i,
\end{eqnarray}
where (a) follows by $i\in supp(\textit{\textbf{w}}^+)$, $\tilde h_i(i)=\gamma$ and the other items in the summation are non-negative.
On the other hand, since $0\leq \tilde{h}_j(i)\leq\lambda$ for any $j\in supp(\textit{\textbf{w}}^-)$,
\begin{eqnarray}\label{pfpart2}
-\sum_{j\in supp(\textit{\textbf{w}}^-)} \tilde{h}_j(i)w_j\leq-\lambda\sum_{j\in supp(\textit{\textbf{w}}^-)}w_j\overset{(b)}{=}\lambda\frac{||\textit{\textbf{w}}||_1}{2}.
\end{eqnarray}
where (b) follows by (\ref{poseqneg}).
Combining (\ref{pfpart1}) with (\ref{pfpart2}), we have that for any $w_i>0$, $w_i\leq\frac{\lambda||\textit{\textbf{w}}||_1}{2\gamma}$.

Similarly, for any $w_i<0$, we have that
\begin{eqnarray*}
-\gamma w_i \le -\!\!\!\!\sum_{j\in supp(\textit{\textbf{w}}^-)}\!\!\tilde{h}_j(i)w_j=\!\!\!\!\sum_{j\in supp(\textit{\textbf{w}}^+)}\!\!\tilde{h}_j(i)w_j \,\le\, \lambda\frac{||\textit{\textbf{w}}||_1}{2},
\end{eqnarray*}
or $-w_i\leq\frac{\lambda||\textit{\textbf{w}}||_1}{2\gamma}$, and this completes the proof.
\end{proof}
\begin{Remark}
Combining Theorem \ref{thg4} with Theorem \ref{preth}, we have that the $l_1$-optimization (\ref{l1}) can exactly reconstruct any $k$-sparse signal with $k<\frac{\gamma}{\lambda}$ when the binary measurement matrix $H$ with uniform column weight $\gamma$ is used.
According to the best known result based on RIP \cite[Prop.1]{jbsd}, $H$ satisfies the RIP of order $s<1+\frac{\gamma}{\lambda}$ and
the $l_1$-optimization (\ref{l1}) can exactly reconstruct $k$-sparse signals with $k<(1+\frac{\gamma}{\lambda})/2$. Thus, our result of Theorem \ref{thg4} improves it significantly in most cases.
\end{Remark}

\begin{Example}
Let $g=4$, $\gamma=4$, $\lambda=2$, and $H$ be the point-plane incidence matrix of a 3-dimensional Euclidean geometry
over $\{0,1\}$ (see \cite{cd} for details), i.e.,
\begin{center}
\begin{eqnarray*}
    H=\left(
  \begin{array}{cccccccccccccc}
    1&0&0&1&1&1&0&0&1&1&0&0&0&1 \\
    0&1&0&0&1&1&1&1&0&1&1&0&0&0 \\
    1&0&1&0&0&1&1&0&1&0&1&1&0&0 \\
    1&1&0&1&0&0&1&0&0&1&0&1&1&0 \\
    1&1&1&0&1&0&0&0&0&0&1&0&1&1 \\
    0&1&1&1&0&1&0&1&0&0&0&1&0&1 \\
    0&0&1&1&1&0&1&1&1&0&0&0&1&0 \\
    0&0&0&0&0&0&0&1&1&1&1&1&1&1
  \end{array}
\right).
\end{eqnarray*}
\end{center}
It is easy to see $\bar{\textit{\textbf{w}}}=(1, -1, 0, 0, 0, 0, 0, 1, -1, 0, 0, 0, 0, 0)^T \in Nullsp_{\mathbb{R}}^{*}(H)$ and $\forall i\in supp(\bar{\textit{\textbf{w}}})$, $|\bar{w}_i|={||\bar{\textit{\textbf{w}}}||_1}/{4}$. Note that $1+\gamma/\lambda=3$, $2\gamma/\lambda=4$ and the bound (\ref{mreg4}) is achieved.
\end{Example}

For a binary matrix with $g\geq6$, $\lambda=1$ and thus $|w_i|\leq\frac{||\textit{\textbf{w}}||_1}{2\gamma}$  by Theorem \ref{thg4}, which is independent of $g$. The next theorem applies to $g>4$ and gives better results for $g>8$.

\begin{Theorem}\label{thg6}
Let $H$ be a binary $m\times n$ measurement matrix with uniform column weight $\gamma\ge 2$ and girth $g\geq6$.
Then for any $\textit{\textbf{w}}\in Nullsp_{\mathbb R}^*(H)$ and any $i\in supp(\textit{\textbf{w}})$,
\begin{equation}\label{mreg6}
    |w_{i}|\leq \frac{||\textit{\textbf{w}}||_{1}}
    {C_0^{(g\geq6)}}\quad \mbox{with}\quad  C_0^{(g\geq6)}\triangleq 2\sum_{u=0}^{t+1}{(\gamma-1)^u},
\end{equation}
where $t\triangleq\lfloor\frac{g-6}{4}\rfloor$ and $\lfloor x\rfloor$ is the floor function which denotes the maximum integer not greater than $x$.
\end{Theorem}

\begin{proof}
Clearly, $g=4t+6$ for odd $g/2$ and
$g=4t+8$ for even $g/2$.
Let $G_{H}$ be the Tanner graph of $H$ with $m$ check nodes and $n$ variable nodes.
For any variable node $i\in supp(\textit{\textbf{w}})$, we construct a local tree of $i$ (see Fig. \ref{local_tree}) as in the proof of \cite[Th. 3.1]{ckds}\cite[Th. 1]{sxff}.

In the local tree of $i$, $i$ is the root
of the tree. A check node $f$ connected to $i$ is called a child of
$i$, and a variable node $j$ connected to $f$ except its parent $i$
is called a child of $f$ or a grandchild of $i$, and a check node
$e$ connected to $j$ except its parent $f$ is called a child of $j$,
and so on. For a variable node $j$ in the local tree, let ${\rm
child}(j)$ and ${\rm grch}(j)$ denote the sets of all children and
grandchildren of $j$ respectively. Note that
$
\label{grch} {\rm grch}(j)=\bigcup_{f\in {\rm child}(j)} {\rm
child}(f).
$
All nodes in $L_0(i)= {\rm grch}(i)$ are \emph{Level}-0
variable nodes. For $u=1,2,\ldots,t$, all nodes in
\begin{eqnarray}
\label{lm} L_u(i)= \bigcup_{j\in L_{u-1}(i)}{\rm grch}(j)
\end{eqnarray}
are \emph{Level}-$u$ variable nodes. Fixing a check node
$f^*\in {\rm child}(i)$, denote $N_0(f^*)={\rm child}(f^*)$ and
\begin{eqnarray}
\label{nm} N_{u}(f^*)=\bigcup_{j\in N_{u-1}(f^*)}{\rm grch}(j),
\;\;\; u=1,2,\ldots, t+1.
\end{eqnarray}
The local tree of $i$ has levels $0$ through $t$ if $g=4t+6$ and $0$
through $t+1$ if $g=4t+8$, where $N_{t+1}(f^*)$ is the set of
$(t+1)$-th level nodes. Since the Tanner graph $G_H$ has girth $g\ge
6$, it is clear that if $g=4t+6$, $\{i\},
L_0(i),\ldots, L_t(i)$ are pairwise disjoint and if $g=4t+8$, $\{i\}, L_0(i),\ldots,L_t(i)$, $N_{t+1}(f^*)$ are pairwise
disjoint.

Since there are $\gamma$ $1$'s in every column of $H$, we have that
$|{\rm child}(i)|=\gamma$ and $|{\rm
child}(j)|=\gamma-1$ for any intermediate variable node $j$.
Since $\textit{\textbf{w}}\in Nullsp_{\mathbb R}^*(H)$ satisfies every check equation in ${\rm child}(i)$, we have
$\gamma w_{i} + \sum_{j\in L_{0}(i)}w_{j}=0$ by adding
these equations. Similarly, $(\gamma-1) w_{j} + \sum_{j'\in {\rm grch}(j)}w_{j'}=0$
for any intermediate variable node $j$. Thus, by using the above two equalities iteratively, we have that
\begin{eqnarray}
  w_{i} &=& (-1)^0\cdot w_{i}, \label{eq1}\\
  \gamma w_{i} &=& (-1)^1\cdot \sum_{j\in L_{0}(i)}w_{j}, \label{eq2}\\
  \gamma (\gamma-1) w_{i} &=& (-1)^2\cdot \sum_{j\in L_{1}(i)}w_{j}, \label{eq3}\\
  \gamma (\gamma-1)^2 w_{i} &=& (-1)^3\cdot \sum_{j\in L_{2}(i)}w_{j}, \label{eq4}\\
  &\vdots& \nonumber\\
  \gamma (\gamma-1)^{t-1} w_{i} &=& (-1)^t\cdot \sum_{j\in L_{t-1}(i)}w_{j}, \label{eq_{t+1}}\\
  \gamma (\gamma-1)^t w_{i} &=& (-1)^{t+1}\cdot \sum_{j\in L_{t}(i)}w_{j}, \label{eq_{t+2}}
\end{eqnarray}
and when $g=4t+8$, we additionally have
\begin{eqnarray}
  (\gamma-1)^{t+1} w_{i} &=& (-1)^{t+2}\cdot \sum_{j\in N_{t+1}(f^*)}w_{j}. \label{eq_{t+3}}
\end{eqnarray}

\begin{figure}[htbp]
\centering
\includegraphics[width=0.4\textwidth]{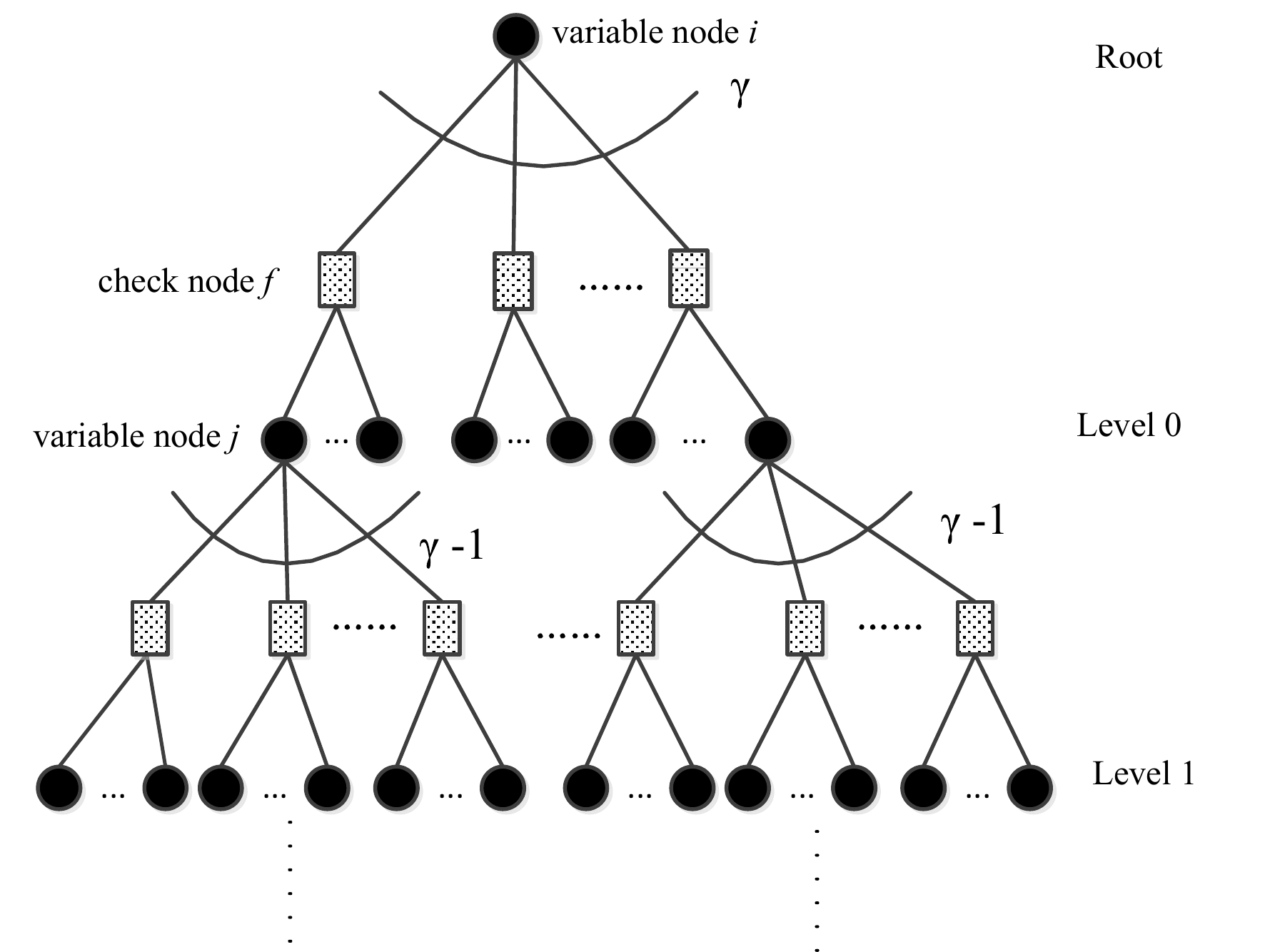}
\begin{center}
  \caption {Local tree of variable node $i$}\label{local_tree}
\end{center}
\end{figure}

The following proofs are divided into 4 cases according to the parities of $g/2$ and $t$.

\textit{Case 1: $g/2$ is odd, $t$ is odd}.

If $w_i>0$, by summing all equations having even power of $(-1)$ on the right, or (\ref{eq1}), (\ref{eq3}), \ldots, and (\ref{eq_{t+2}}), we have
\begin{eqnarray}
&&\left[1+\sum_{s=0}^{(t-1)/2}\gamma(\gamma-1)^{2s+1}\right]w_i\nonumber\\
&=&w_i+\sum_{s=0}^{(t-1)/2}\sum_{j\in L_{2s+1}(i)}{w_j}\nonumber\\
&\le &w_i+\sum_{s=0}^{(t-1)/2}\sum_{j\in L_{2s+1}(i),\;w_j>0}{w_j}\nonumber\\
&\le &w_i+\sum_{u=0}^{t}\sum_{j\in L_{u}(i),\;w_j>0}{w_j}\nonumber\\
&\overset{(a)}{\leq}& \sum_{j\;:\;w_{j}>0}{w_{j}} \;
\overset{(b)}{=}\;\frac{||\textit{\textbf{w}}||_{1}}{2},\label{pos0}
\end{eqnarray}
where (a) holds because $\{i\},
L_0(i),\ldots, L_t(i)$ are pairwise disjoint, and (b) follows by (\ref{poseqneg}).

If $w_i<0$, we have
\begin{eqnarray}
  &&\left[1+\sum_{s=0}^{(t-1)/2}\gamma(\gamma-1)^{2s+1}\right]w_i\nonumber\\
&=&w_i+\sum_{s=0}^{(t-1)/2}\sum_{j\in L_{2s+1}(i)}{w_j}\nonumber\\
&\ge &w_i+\sum_{s=0}^{(t-1)/2}\sum_{j\in L_{2s+1}(i),\;w_j<0}{w_j}\nonumber\\
&\ge & \sum_{j\;:\;w_{j}<0}{w_{j}} \;
{=}\;-\frac{||\textit{\textbf{w}}||_{1}}{2}.\label{pos0n}
\end{eqnarray}

Therefore, combining (\ref{pos0}) with (\ref{pos0n}), we have that
\begin{eqnarray}\label{part0}
|w_{i}|&\leq&\frac{||\textit{\textbf{w}}||_{1}}
    {2\left[1+\sum_{s=0}^{(t-1)/2}\gamma(\gamma-1)^{2s+1}\right]}\nonumber\\
    &\overset{(c)}{=}&\frac{||\textit{\textbf{w}}||_{1}}{2\sum_{u=0}^{t+1}(\gamma-1)^u},
\end{eqnarray}
where (c) follows from $\gamma(\gamma-1)^{v}=(\gamma-1)^{v}+(\gamma-1)^{v+1}$.

By summing all equations having odd power of $(-1)$ on the right, or (\ref{eq2}), (\ref{eq4}), \dots, and (\ref{eq_{t+1}}), we could similarly
obtain
\begin{eqnarray}
|w_i|
&\leq& \frac{||\textit{\textbf{w}}||_{1}}{2\sum_{s=0}^{(t-1)/2}\gamma(\gamma-1)^{2s}}. \label{neg0}
\end{eqnarray}
But it is weaker than (\ref{part0}) since
\begin{eqnarray*}
  1+\sum_{s=0}^{(t-1)/2}\gamma(\gamma-1)^{2s+1}&>&\sum_{s=0}^{(t-1)/2}\gamma(\gamma-1)^{2s}.
\end{eqnarray*}

\emph{Case 2: $g/2$ is odd, $t$ is even.}

With totally similar arguments, we have
\begin{eqnarray}\label{pos1}
  |w_i|&\leq&\frac{||\textit{\textbf{w}}||_{1}}{2\left[1+\sum_{s=1}^{t/2}\gamma(\gamma-1)^{2s-1}\right]}
\end{eqnarray}
and
\begin{eqnarray}\label{neg1}
     |w_{i}|\leq\frac{||\textit{\textbf{w}}||_{1}}{2\sum_{s=0}^{t/2}\gamma(\gamma-1)^{2s}}.
\end{eqnarray}
Since
\begin{eqnarray*}
  \sum_{u=0}^{t+1}(\gamma-1)^u=\sum_{s=0}^{t/2}\gamma(\gamma-1)^{2s}>1+\sum_{s=1}^{t/2}\gamma(\gamma-1)^{2s-1},
\end{eqnarray*}
we also have
\begin{eqnarray}\label{part1}
    |w_{i}|&\leq& \frac{||\textit{\textbf{w}}||_{1}}
    {2\sum_{u=0}^{t+1}(\gamma-1)^u}.
\end{eqnarray}

\textit{Case 3: $g/2$ is even, $t$ is odd}.

This case is different from Case 1 in that while summing the equations with odd power of $(-1)$ on the right, an extra (\ref{eq_{t+3}}) needs to be summed.
By summing (\ref{eq2}), (\ref{eq4}), \dots, (\ref{eq_{t+1}}) and (\ref{eq_{t+3}}), we could similarly obtain
\begin{eqnarray}
  |w_{i}|&\leq&\frac{||\textit{\textbf{w}}||_{1}}{2\left[\sum_{s=0}^{(t-1)/2}\gamma(\gamma-1)^{2s}+(\gamma-1)^{t+1}\right]}\nonumber\\
  &=& \frac{||\textit{\textbf{w}}||_{1}}{2\sum_{u=0}^{t+1}(\gamma-1)^u}.\label{part2}
\end{eqnarray}
Note that the two summing methods give identical results.

\emph{Case 4: $g/2$ is even, $t$ is even}.

This case is different from Case 2 in that while summing
the equations having even power of $(-1)$ on the right, an extra (\ref{eq_{t+3}}) needs to be summed.
By summing (\ref{eq1}), (\ref{eq3}), \ldots, (\ref{eq_{t+1}}) and (\ref{eq_{t+3}}), we could similarly obtain
\begin{eqnarray}
  |w_{i}|&\leq&\frac{||\textit{\textbf{w}}||_{1}}{2\left[1+\sum_{s=1}^{t/2}\gamma(\gamma-1)^{2s-1}+(\gamma-1)^{t+1}\right]}\nonumber\\
  &=& \frac{||\textit{\textbf{w}}||_{1}}{2\sum_{u=0}^{t+1}(\gamma-1)^u}.\label{part3}
\end{eqnarray}
Note that the two summing methods give identical results.

Combining (\ref{part0}) with (\ref{part1})--(\ref{part3}), the proof is completed.
\end{proof}

\begin{Remark}\label{comp}
Combining Theorem \ref{thg6} with Theorem \ref{preth}, we have that the $l_1$-optimization (\ref{l1}) can exactly reconstruct any $k$-sparse signal with $k<\sum_{u=0}^{t+1}{(\gamma-1)^u}$ when the binary measurement matrix $H$ with uniform column weight $\gamma\ge 2$ and girth $g\ge 6$ is used.
Kelley and Sridhara \cite{ckds} showed that
\begin{eqnarray}\label{pseudo}
  &&w_p^{BSC,min}(H)\nonumber\\
  &\geq&\left\{ \begin{array}{ll}
  2\sum_{u=0}^{t+1}(\gamma-1)^u-(\gamma-1)^{t+1},\; g/2 \;\;\mbox{is odd},\\
  2\sum_{u=0}^{t+1}(\gamma-1)^u,\quad \quad\quad \quad\quad\;\;\; g/2 \;\;\mbox{is even}.
\end{array}
\right.
\end{eqnarray}
By the result based on NSP \cite[Lem.12, Th.3]{adrs}, the $l_1$-optimization (\ref{l1}) can exactly reconstruct $k$-sparse signals with
\begin{eqnarray}\label{pseudo1}
  k\;<\;
  \left\{ \begin{array}{ll}
  \sum_{u=0}^{t+1}(\gamma-1)^u-\frac{(\gamma-1)^{t+1}}{2},\quad g/2 \;\;\mbox{is odd},\\
  \sum_{u=0}^{t+1}(\gamma-1)^u,\quad \quad\quad \quad\quad\;\;\; g/2 \;\;\mbox{is even}.
\end{array}
\right.
\end{eqnarray}
Thus, our result of Theorem \ref{thg6} improves it when $g/2$ is odd.
\end{Remark}

The following examples show that the bound in (\ref{mreg6}) could be achieved for $g=6,8,10,12$ respectively.

\begin{Example}
Let $g=6$, $\gamma=2$, and $H$ be the point-line incidence matrix of a Euclidean plane over $\{0,1\}$ \cite{cd}, i.e.,
\begin{eqnarray*}
H=\left(
  \begin{array}{cccccc}
    1&0&1&1&0&0 \\
    1&1&0&0&1&0 \\
    0&1&1&0&0&1 \\
    0&0&0&1&1&1 \\
  \end{array}
\right).
\end{eqnarray*}
Let $\bar{\textit{\textbf{w}}}=(1, 0, -1, 0, -1, 1)^T\in Nullsp_{\mathbb R}^*(H)$, $\forall i\in supp(\bar{\textit{\textbf{w}}})$, $|\bar w_i|={||\bar{\textit{\textbf{w}}}||_1}/{4}$, which meets the bound (\ref{mreg6}).
\end{Example}

\begin{Example}
Let $g=8$, $\gamma=2$, and $H$ be the point-line incidence matrix of a cube (see \cite{cd} for details), i.e.,
\begin{eqnarray*}\label{cubematrix}
H=\left(
  \begin{array}{cccccccccccc}
    1&0&0&1&0&0&0&0&1&0&0&0\\
    1&1&0&0&0&0&0&0&0&1&0&0\\
    0&1&1&0&1&0&0&0&0&0&0&0\\
    0&0&1&1&0&0&1&0&0&0&0&0\\
    0&0&0&0&0&0&0&1&1&0&0&1\\
    0&0&0&0&0&0&0&0&0&1&1&1\\
    0&0&0&0&1&1&0&0&0&0&1&0\\
    0&0&0&0&0&1&1&1&0&0&0&0\\
  \end{array}
\right).
\end{eqnarray*}
Let $\bar{\textit{\textbf{w}}}=(1, -1, 1, -1, 0, 0, 0, 0, 0, 0, 0, 0)^T\in Nullsp_{\mathbb R}^*(H)$, $\forall i\in supp(\bar{\textit{\textbf{w}}})$, $|\bar{w}_i|={||\bar{\textit{\textbf{w}}}||_1}/{4}$, which meets the bound (\ref{mreg6}).
\end{Example}

\begin{Example}
Consider the point-line incidence matrix $H$ of $GP(5, 2)$ \cite{cd}, $g=10$, $\gamma=2$,
\begin{center}
\begin{eqnarray*}\label{pgmatrix}
H\!=\!\!\left(\!\!
  \begin{array}{ccccccccccccccc}
  1&0&0&0&1&1&0&0&0&0&0&0&0&0&0\\
  1&1&0&0&0&0&0&0&0&1&0&0&0&0&0\\
  0&1&1&0&0&0&0&0&1&0&0&0&0&0&0\\
  0&0&1&1&0&0&0&1&0&0&0&0&0&0&0\\
  0&0&0&1&1&0&1&0&0&0&0&0&0&0&0\\
  0&0&0&0&0&1&0&0&0&0&1&0&0&0&1\\
  0&0&0&0&0&0&0&0&0&1&0&0&1&1&0\\
  0&0&0&0&0&0&0&0&1&0&1&1&0&0&0\\
  0&0&0&0&0&0&0&1&0&0&0&0&0&1&1\\
  0&0&0&0&0&0&1&0&0&0&0&1&1&0&0
  \end{array}
\!\!\right).
\end{eqnarray*}
\end{center}
Let $\bar{\textit{\textbf{w}}}=(-1,1,0,0,1,0,-1,0,-1,0,0,1,0,0,0)^T$, then $\bar{\textit{\textbf{w}}}\in
Nullsp_{\mathbb R}^*(H)$ and it meets the bound (\ref{mreg6}).
\end{Example}

\begin{Example}
Consider the following $H$ with $g=12$, $\gamma=2$,
\begin{center}
\begin{eqnarray*}
H=\left(
  \begin{array}{cccccccccccc}
  1&1&1&0&0&0&0&0&0&0&0&0 \\
  0&0&0&1&1&1&0&0&0&0&0&0 \\
  0&0&0&0&0&0&1&1&1&0&0&0  \\
  0&0&0&0&0&0&0&0&0&1&1&1  \\
  1&0&0&1&0&0&0&0&0&0&0&0\\
  0&0&0&0&1&0&1&0&0&0&0&0  \\
  0&1&0&0&0&0&0&1&0&0&0&0  \\
  0&0&1&0&0&0&0&0&0&1&0&0   \\
  0&0&0&0&0&1&0&0&0&0&1&0  \\
  0&0&0&0&0&0&0&0&1&0&0&1
  \end{array}
\right)
\end{eqnarray*}
\end{center}
Let $\bar{\textit{\textbf{w}}}=(1,0,-1,-1,0,1,0,0,0,1,-1,0)^T\in Nullsp_{\mathbb R}^*(H)$, $\forall i\in supp(\bar{\textit{\textbf{w}}})$, $|\bar w_i|={||\bar{\textit{\textbf{w}}}||_1}/{6}$, which meets the bound (\ref{mreg6}).
\end{Example}
\section{Conclusion}\label{conclusion}
This paper has considered the performance of binary measurement matrices with uniform column weight under basis pursuit.
Girth of such matrices is employed to provide sharp $l_1/l_1$, $l_2/l_1$ and $l_\infty/l_1$ sparse approximation guarantees for $l_1$-optimizations, which improve previous known RIP results for any girth and NSP results for girth $4$ and $4t+6$.
Our results show that the larger the girth and/or column weight is, the better the binary measurement matrix will be, which
further shows that large girth has positive impacts on the performance of binary measurement matrices. Our results and methods have close relations with that in LDPC codes, which suggests that the parity-check matrices of good LDPC codes are important candidates of measurement matrices.

For fixed $n$, $g$ and $\gamma$, there seems to be a minimum $m$ for any binary matrix.
In the future, we will explore this and it will help us to check that whether binary matrices can overcome $k=O(\sqrt{m})$ \cite{lggz} in the sense of ``strong bounds" \cite{adrs}.

\section*{Acknowledgment}
The authors wish to express their appreciation to the anonymous reviewers for their valuable suggestions and comments that helped to greatly improve the paper.

This research is supported in part by the 973 Program of China (2012CB315803), the Research Fund for the Doctoral Program of Higher Education of China (20100002110033), and the open research fund of National Mobile Communications Research Laboratory of Southeast University (2011D11).
Shu-Tao Xia is also with the National Mobile Communications Research Laboratory of Southeast University of China.

\end{document}